\documentclass[11pt]{article}
\usepackage{amsmath,amsthm,amssymb}
\usepackage[usenames,dvipsnames]{xcolor}
\usepackage{url}
\usepackage{geometry}
\usepackage{graphicx,float,epstopdf}
\restylefloat{figure}

\newtheorem{theorem}{Theorem}
\newtheorem{corollary}[theorem]{Corollary}
\newtheorem{lemma}[theorem]{Lemma}
\newtheorem{proposition}[theorem]{Proposition}

\theoremstyle{definition}

\newtheorem{remark}[theorem]{Remark}
\newtheorem{example}[theorem]{Example}

\newtheorem{construction}[theorem]{Construction}
\newtheorem{algorithm}[theorem]{Algorithm}

\newcommand{\B}{\mathcal{B}}

\newcommand{\BIBD}{\mathrm{BIBD}}

\newcommand{\PBD}{\mathrm{PBD}}

\newcommand{\supp}{\textrm{supp}}
\newcommand{\RIP}{\mathrm{RIP}}

\newcommand{\op}[1]{\|#1\|_{\mathrm{op}}}

\DeclareMathOperator{\PG}{\textnormal{PG}}

\title{\textbf{\Large{Compressed sensing with combinatorial designs: theory and simulations}}}

\author{
\textsc{Darryn Bryant}
				\thanks{\textit{E-mail: db@maths.uq.edu.au}}\\
\textit{\footnotesize{School of Mathematics and Physics,}}\\
\textit{\footnotesize{University of Queensland, QLD 4072, Australia.}}\\
\textsc{Charles J. Colbourn}
                \thanks{\textit{E-mail: charles.colbourn@asu.edu}}\\
\textit{\footnotesize{School of Computing, Informatics and Decision Systems Engineering,}}\\
\textit{\footnotesize{Arizona State University, Tempe 85287-8809, Arizona, U.S.A.}}\\
\textsc{Daniel Horsley}
\               \thanks{\textit{E-mail: daniel.horsley@monash.edu}}\\
\textit{\footnotesize{School of Mathematical Sciences,}}\\
\textit{\footnotesize{Monash University, VIC 3800, Australia.}}\\
\textsc{Padraig \'O Cath\'ain}
				\thanks{\textit{E-mail: p.ocathain@gmail.com}}\\
\textit{\footnotesize{School of Mathematical Sciences,}}\\
\textit{\footnotesize{Monash University, VIC 3800, Australia.}}\\
}

\begin{document}

\maketitle
\begin{center}
\begin{abstract}
\noindent
In \textit{An asymptotic result on compressed sensing matrices}, a new construction for compressed sensing matrices using combinatorial design theory was introduced. In this paper, we use deterministic and probabilistic methods to analyse the performance of matrices obtained from this construction.
We provide new theoretical results and detailed simulations.
These simulations indicate that the construction is competitive with Gaussian random matrices, and that recovery is tolerant to noise.
A new recovery algorithm tailored to the construction is also given.
\end{abstract}

\end{center}

\vspace{0.3cm}



\clearpage

\section{Overview}

In 2006 Donoho, Cand\`{e}s, Romberg and Tao
\cite{CandesRombergTaoRobustUncertainty, CandesRombergTaoStableRecovery,DonohoCS}
laid the foundations for a revolutionary new signal sampling paradigm, which is now called \textit{compressed sensing}. Their key insight was that many real-world signals have the special property of being \textit{sparse} -- they can be stored much more concisely than a random signal. Instead of sampling the whole signal and then applying data compression algorithms, they showed that sampling and compression of sparse signals can be achieved simultaneously. This process requires dramatically fewer measurements than the number dictated by traditional thinking, but requires complex measurements that are `incoherent' with respect to the signal.

Cand\`{e}s, Romberg and Tao established fundamental constraints for
sparse recovery; one cannot hope to recover $k$-sparse signals of length $N$ in less than $O(k\log N)$
measurements under any circumstances. They then established that several classes of random matrices
meet this bound asymptotically. Two important examples are the \textit{random Gaussian ensemble},
which has entries drawn from a standard normal distribution, and the \textit{random Fourier ensemble},
which consists of a random selection of rows from the discrete Fourier transform matrix \cite{CandesRombergTaoRobustUncertainty,CandesRombergTaoStableRecovery}.

In \cite{mypaper-PBD}, a construction for compressed sensing matrices based on block designs and complex Hadamard matrices
was introduced (see Construction \ref{con0} below).
Here we add to the analysis of these matrices, both establishing new results
on their compressed sensing performance, and providing details of extensive simulations. Our main results are the following.

\begin{enumerate}

    \item Theorem 10 of \cite{mypaper-PBD} establishes
that an $n\times N$ matrix created via Construction \ref{con0}
has the $(\ell_{1}, t)$-recovery property for all values of $t \leq \frac{\sqrt{n}}{4}$, where $n$ is the number of rows in the matrix. In Section \ref{performanceBoundsSec} we show that there exist vectors of sparsity at most $\sqrt{2n}$ which cannot be recovered. Recall that $\Phi$ has the $(\ell_{1}, t)$-recovery property if every $t$-sparse vector $m$ is uniquely recoverable from its image $\Phi m$ by $\ell_{1}$-minimisation.

    \item In Section \ref{heuristicsSec} we give a non-rigorous analysis of Construction \ref{con0} which suggests that sufficiently large matrices created via the construction allow recovery of most vectors of sparsity $O(\sqrt{n\log n})$.

    \item In Section \ref{simulationsSec} we provide detailed simulations which suggest that the recovery performance of matrices obtained from Construction \ref{con0} is comparable to that of the Gaussian ensemble for matrices with hundreds of rows and thousands of columns. Simulations also suggest that signal recovery is robust against uniform and burst noise. Algorithm running times are better than for Gaussian matrices.

    \item In Section \ref{NewAlgSec} we propose a new algorithm for signal recovery, tailored to Construction \ref{con0}. We show that for most vectors of sparsity at most $\sqrt{n}$, this algorithm runs in time $O(n\log n)$.
\end{enumerate}

\section{Preliminaries}

For our purposes here, a {\em Hadamard matrix of order $r$} is an $r \times r$ matrix $H$ such that each entry of $H$ is a complex number with magnitude 1 and $HH^*=rI_r$ (where $H^{\ast}$ is the conjugate transpose). The prototypical example of such a matrix is the character table of an abelian group; for cyclic groups this matrix is often referred to as a \textit{discrete Fourier transform matrix} or simply a \textit{Fourier matrix}.
When we restrict the entries to real numbers, we state explicitly that the matrix is a real Hadamard matrix. For background on Hadamard matrices, we refer the reader to \cite{HoradamHadamard}.

A {\em pairwise balanced design} $(V,\B)$ consists of a set $V$ of {\em points} and a collection $\B$ of subsets of $V$, called {\em blocks}, such that each pair of points is contained in exactly $\lambda$ blocks for some fixed positive integer $\lambda$. If $v=|V|$ and $K$ is a finite set of integers such that $|B| \in K$ for each $B \in \B$, we use the notation $\PBD(v,K,\lambda)$. We denote the maximum element of $K$ by $K_{\max}$ and the minimum element of $K$ by $K_{\min}$. For each point $x\in V$, the {\em replication number} $r_x$ of $x$ is defined by $r_x=|\{B\in\B:x\in B\}|$.
A set of points in a PBD is an {\em arc} if at most two of the points occur together in any block.

A $\PBD(v,\{k\},\lambda)$ is a {\em balanced incomplete block design}, denoted by $\BIBD(v,k,\lambda)$. Obviously, all points of a BIBD must have the same replication number. This paper is devoted to the study of matrices created via
Construction \ref{con0}, which generalises a construction from \cite{Mixon}.

\begin{construction}[\cite{mypaper-PBD}]\label{con0}
Let $(V,\B)$ be a $\PBD(v,K,1)$ with $n = |\B|$ and $N = \sum_{x \in V} r_{x}$.
Then $\Phi$ is the $n\times N$ matrix constructed as follows.
\begin{itemize}
    \item Let $A$ be the transpose of the incidence matrix of $(V,\B)$; rows of $A$
    are indexed by blocks, columns of $A$ by points, and the entry in row $B$ and
    column $x$ is $1$ if $x\in B$ and $0$ otherwise.
    \item For each $x \in V$, let $H_{x}$ be a (possibly complex) Hadamard matrix of order $r_{x}$.
    \item For each $x \in V$, column $x$ of $A$ determines $r_{x}$
    columns of $\Phi$; each zero in column $x$ is replaced with the $1 \times r_{x}$
    row vector $(0,0,\ldots,0)$, and each $1$ in column $x$ is replaced with a distinct row of $\frac{1}{\sqrt{r_{x}}}H_{x}$.
\end{itemize}
\end{construction}

\begin{remark}\label{growthRemark}
Suppose that an $n\times N$ matrix is created via Construction \ref{con0} using a $\PBD(v,K,1)$ $\mathcal{D}$ with replication numbers $r_1,\ldots,r_v$ and block sizes $k_1,\ldots,k_n$. We briefly discuss the behaviour of $n$ and $N$ in terms of the parameters of $\mathcal{D}$ in the case where $v$ is large and $K_{\min} \sim K_{\max}$, that is, there exists a constant $\alpha$ not depending on $v$ such that $K_{\max} \leq \alpha K_{\min}$. Let $r=\frac{1}{v}(r_1+\cdots+r_v)$ and $k=\frac{1}{n}(k_1+\cdots+k_n)$ be the average replication number and average block size respectively. Obviously $K_{\min},K_{\max} \sim k$, and standard counting arguments for block designs (see Chapter II of \cite{BJL} for example) yield
\begin{gather}
  N=rv=kn \label{exactN}\\
  kr \sim v \label{vlinsim}\\
  k^2n \sim v^2 \label{vquadsim}
\end{gather}
Combining \eqref{vlinsim} and \eqref{vquadsim}, we see that $n \sim r^2$. It is known (see \cite{deBruijnErdos}) that any non-trivial PBD has at least as many blocks as points. So $n \geq v$, and hence $k$ grows no faster than $O(\sqrt{v})$ by \eqref{vquadsim}. Two extremes can be distinguished. When $k$ is constant, $v \sim r \sim \sqrt{n}$ and $N \sim r^2 \sim n$ -- this is the case considered in \cite{mypaper-PBD}. When $k\sim\sqrt{v}$, $v \sim r^2 \sim n$ and $N \sim r^3 \sim n^{\frac{3}{2}}$ (this occurs when $\mathcal{D}$ is a projective plane, for example). It is easy to see that the growth rates of $v$ and $N$ are bounded by these two extremes.
\end{remark}

Except when $n \in \{1,2\}$, there do not exist real Hadamard matrices of order $n $ for $n \not\equiv 0 \pmod{4}$. If complex Hadamard matrices are employed in Construction \ref{con0},  the resulting compressed sensing matrix also has complex entries. If real matrices are required for a particular application (for example, linear programming solvers generally require the ordered property of the real numbers),  the isometry in Lemma \ref{BourgainTrick} can be exploited. While this result is elementary and well known, we first became aware of its use in this context in \cite{Bourgain}.

\begin{lemma}\label{BourgainTrick}
The map
\[a + ib \mapsto \left(\begin{array}{rr} a & b \\ -b & a \end{array}\right) \]
is an isometry from $\mathbb{C}$ to $M_{2}(\mathbb{R})$.
It extends to an isometry from $M_{n,N}(\mathbb{C})$ to $M_{2n,2N}(\mathbb{R})$.
\end{lemma}

This technique allows a complex compressed sensing matrix to be converted into a real one,
at the expense of doubling its dimensions. If a complex matrix recovers $t$-sparse vectors,
so too does the corresponding real matrix. (But while a $t$-sparse complex vector
is in general $2t$-sparse, we cannot, and do not claim to, recover arbitrary $2t$-sparse real vectors.)
Where confusion may arise, we specify the field over which our matrix is defined as a subscript;
thus if $\Phi_{\mathbb{C}}$ is $n\times N$, $\Phi_{\mathbb{R}}$ is $2n \times 2N$.

We  require a result concerning the sparsity of linear
combinations of rows of a complex Hadamard matrix, which is essentially
an uncertainty principle bounding how well the standard normal basis can
be approximated by the basis given by the (scaled) columns of a Hadamard matrix.
In particular, if $m \in \mathbb{C}^{n}$ admits an expression as a linear combination
of $u$ columns of a Hadamard matrix, and as a linear combination of $d$ standard basis
vectors, then $du \geq n$.

\begin{lemma}[\cite{mypaper-switching}, Lemma 2, cf. \cite{Jukna}, Lemma 14.6]\label{rowbound}
Let $H$ be a complex Hadamard matrix of order $n$. If $m$ is a non-zero linear combination of at most $u$ columns of $H$, then $m$ has at least
$\lceil \frac{n}{u}\rceil$ non-zero entries.
\end{lemma}

The bound in Lemma \ref{rowbound} is sharp. Let $H$ be a Fourier matrix of order $v$,
and suppose that $u$ divides $v$. Then there exist $u$ columns of $H$ containing
only $u^{\textrm{th}}$ roots of unity and the sum of these columns vanishes on all
but $\frac{v}{u}$ coordinates. At the other extreme,  a non-trivial linear combination of any $u$ of the columns of  a Fourier matrix of prime order vanishes on at most $u-1$ coordinates.

\section{Upper bounds on performance}\label{performanceBoundsSec}

The \textit{spark} of a matrix $\Phi$ is the smallest non-zero value of $s$ such that
there exists a vector $m$ of sparsity $s$ in the nullspace of $\Phi$. In this section
we provide upper and lower bounds on the spark of a matrix obtained from
Construction \ref{con0}.

The following well-known result bounds the  $(\ell_{1}, t)$-recoverability of a matrix in terms of its spark.

\begin{proposition}\label{sparkRIP}
If the spark of a matrix $\Phi$ is $s$ and $\Phi$ has $(\ell_{1}, t)$-recoverability,
then $t < \frac{s}{2}$.
\end{proposition}

\begin{proof}
Let $\Phi$ be a matrix with spark $s$ and let $m$ be an element of sparsity $s$ in
the nullspace of $\Phi$. Write $m = m_{1} + m_{2}$ where $m_{1}$ has sparsity
$\lfloor \frac{s}{2}\rfloor$ and $m_{2}$ has sparsity $\lceil \frac{s}{2}\rceil$.
Then
\[ \Phi m_{1} + \Phi m_{2} = 0. \]
So $\Phi(m_{2}) = \Phi(-m_{1})$, and one of $-m_1$ or $m_{2}$ is not recoverable.
Thus $\Phi$ does not have $(\ell_{1}, \lceil \frac{s}{2}\rceil)$-recoverability and the result follows.
\end{proof}

We can now provide upper and lower bounds on the spark of $\Phi$.

\begin{proposition}\label{lowerbound}
Let $\mathcal{D}$ be a $\PBD(v,K,1)$ whose smallest replication number is $r_1$ and let $\Phi$ be a matrix obtained from Construction \ref{con0} using $\mathcal{D}$. Then the spark $s$ of $\Phi$ satisfies $r_{1} \leq s$.
\end{proposition}

\begin{proof}
Suppose that $m$ is in the nullspace of $\Phi$.
For each $i \in  \{1,\ldots,v\}$, let $m_i$ be the vector
which is equal to $m$ on those coordinates corresponding to
point $i$ of $\mathcal{D}$ and has each other coordinate equal to 0.
Let $T=\{i:\supp(m_i) \neq \emptyset\}$, and $t=|T|$. Now,
$\sum_{i \in T}\Phi m_i =\Phi m=0$ and, because any two points
of $\mathcal{D}$ occur together in exactly one block,
$|\supp(\Phi m_i) \cap  \supp(\Phi m_j)| \leq 1$ for any distinct
$i,j \in T$. Thus it must be that $|\supp(\Phi m_i)| \leq t-1$ for each $i \in T$.
By Lemma \ref{rowbound}, this implies that $|\supp(m_i)| \geq \frac{r_i}{t-1}$
for each $i \in T$. So, because $r_i \geq r_1$ for each $i \in T$, we have that
$|\supp(m)| = \sum_{i \in T}|\supp(m_i)| \geq \frac{tr_1}{t-1} > r_1$. \qedhere
\end{proof}

\begin{proposition}\label{upperbound}
Let $\mathcal{D}$ be a $\PBD(v,K,1)$ whose two smallest replication numbers are $r_1$ and $r_2$ (possibly $r_1=r_2$) and let $\Phi$ be a matrix obtained from Construction \ref{con0} using $\mathcal{D}$. Then the spark $s$ of $\Phi$ satisfies $s \leq r_{1} + r_{2}$.
\end{proposition}

\begin{proof}
Let $p_{1}$ and $p_{2}$ be points of $\mathcal{D}$ with replication numbers $r_{1}$ and $r_{2}$ respectively, and let $\Phi_1$ and $\Phi_2$ be the submatrices of $\Phi$ consisting of the columns corresponding to $p_1$ and $p_2$ respectively. We  establish the result by showing that the columns of the block matrix $[\Phi_1\ \Phi_2]$ are linearly dependent.

There is a unique row $\rho$ of the matrix $[\Phi_1\ \Phi_2]$ that contains no zero entries (corresponding to the unique block that contains $p_1$ and $p_2$). For $i \in \{1,2\}$, let $\Phi'_i$ be the matrix obtained by scaling the columns of $\Phi_i$ so that all entries in $\rho$ have the value $\frac{1}{r_{1}}$ if $i=1$ and $-\frac{1}{r_{2}}$ if $i=2$. For $i \in \{1,2\}$, because the rows of $\Phi_i$ are pairwise orthogonal, the rows of $\Phi'_i$ are too. In particular, row $t$ of $\Phi_i$ is orthogonal to each other non-zero row of $\Phi_i$ and thus the sum of any row of $\Phi'_i$ other than row $t$ is zero. Clearly, the sum of row $\rho$ of $\Phi_1$ is $1$ and the sum of row $\rho$ of $\Phi_2$ is $-1$. Thus the columns of $[\Phi'_1\ \Phi'_2]$ add to the zero vector, completing the proof.
\end{proof}

The bound of Proposition \ref{upperbound} is sharp. Let $\mathcal{D}$ be a $\PBD(v,K,1)$ with every replication number prime, and let $\Phi$ be a matrix obtained from Construction \ref{con0} taking $(V,\mathcal{B})$ to be $\mathcal{D}$ and $H_x$ to be a Fourier matrix of order $r_x$ for each $x \in V$. Using the fact that a non-trivial linear combination of $u$ columns of a prime order Fourier matrix vanishes on at most $u-1$ coordinates, it can be shown that the spark of $\Phi$ is exactly $r_1+r_2$.
We say that a Hadamard matrix is \textit{optimal} if no linear combination of $k$ rows has more than $k$ zero entries. By this criterion the Fourier matrices of prime order are optimal. The next remark shows that, in general, the spark of $\Phi$ depends both on the choices of the design and the Hadamard matrices used in Construction \ref{con0}.

\begin{remark}\label{counter}
Let $\mathcal{D}$ be a $\PBD(v,K,1)$ with every replication number equal to some
$r \equiv 0 \pmod{4}$, and let $\Phi$ be a matrix obtained from Construction \ref{con0}
taking $(V,\mathcal{B})$ to be $\mathcal{D}$ and $H_x$ to be a real Hadamard matrix
of order $r$ for each $x \in V$. Denote by $\textbf{1}$ a vector of 1s of length
$\frac{r}{2}$. Consider the submatrix $\Phi'$ of $\Phi$ consisting of the columns
corresponding to three points of an arc of $\mathcal{D}$. Then by scaling
and reordering columns of $\Phi'$ we can obtain a matrix $\Phi''$ that has a submatrix of the form

\[ \left(\begin{array}{rr|rr|rr}
\textbf{1} & \textbf{1} & \textbf{1} & \textbf{1} & \textbf{0}& \textbf{0}\\
\textbf{1} & \textbf{-1} & \textbf{0} & \textbf{0} & \textbf{1}& \textbf{1}\\
\textbf{0} & \textbf{0} & \textbf{1} & \textbf{-1} & \textbf{1}& \textbf{-1}\\
\end{array} \right). \]

The vector $m = (\textbf{0}, \textbf{1}, \textbf{0}, \textbf{-1}, \textbf{0}, \textbf{1})^{\top}$ is in the nullspace of $\Phi''$. This is easily verified for the displayed rows, and follows for the remaining rows from the orthogonality of Hadamard matrices (the vector $m$ is a linear combination of the displayed rows of the Hadamard matrix when restricted to any set of columns corresponding to a point). It follows that the nullspace of $\Phi$ contains elements of sparsity $\frac{3r}{2}$.
\end{remark}

Real Hadamard matrices are never optimal, by Remark \ref{counter}. The obvious questions here concern the existence of optimal or near optimal complex Hadamard matrices; some discussion is contained in \cite{mypaper-switching}.

Now we combine our results so far to determine an upper bound
for the $(\ell_{1}, t)$-recoverability of $\Phi$. Recall that MIP-based results are sufficient conditions for
$(\ell_{1}, t)$-recoverability, but are not necessary in general. Thus, while the Welch bound forms a fundamental obstacle to using MIP to prove that vectors of sparsity exceeding $\frac{\sqrt{n}}{2}$ sampled by an $n \times N$ matrix can be recovered,
it does \textbf{not} guarantee that such vectors cannot be recovered. We show  that there exist vectors of
sparsity at most $\sqrt{2n}$ which cannot be recovered by matrices obtained via Construction \ref{con0}.

\begin{proposition}\label{CSUpperBound}
Let $\mathcal{D}$ be a $\PBD(v,K,1)$ such that $K_{\max} \leq \sqrt{2}(K_{\min}-1)$ and the
two smallest replication numbers of $\mathcal{D}$ are $r_1$ and $r_2$ (possibly $r_1=r_2$).
Let $\Phi$ be an $n\times N$ matrix obtained from Construction \ref{con0} using $\mathcal{D}$.
If $\Phi$ has $(\ell_{1}, t)$-recoverability, then $t < \sqrt{2n}$.
\end{proposition}

\begin{proof}
By Proposition \ref{upperbound}, the spark $s$ of $\Phi$ is at most $r_{1}+r_{2}$.
The replication number of any point is at most $\frac{v-1}{K_{\min}-1}$, so
\[s \leq \frac{2(v-1)}{K_{\min}-1}. \]
Because $n$ is the number of blocks in $\mathcal{D}$ and $K_{\max} \leq \sqrt{2}(K_{\min}-1)$, we have
\[n \geq \frac{v(v-1)}{K_{\max}(K_{\max}-1)} > \frac{(v-1)^{2}}{K_{\max}^2} \geq \frac{(v-1)^{2}}{2(K_{\min}-1)^2} \geq \frac{s^2}{8} \]
and thus $s < 2 \sqrt{2n}$. The result now follows by Proposition \ref{sparkRIP}.
\end{proof}

Combining Proposition \ref{CSUpperBound} with Theorem 10 of \cite{mypaper-PBD} we can specify the $(\ell_{1}, t)$-recoverability of a matrix obtained from Construction \ref{con0} to within a multiplicative factor of $4\sqrt{2}$.

\begin{theorem}\label{CSBounds}
Let $\mathcal{D}$ be a $\PBD(v,K,1)$ such that $K_{\max} \leq \sqrt{2}(K_{\min}-1)$
and the two smallest replication numbers of $\mathcal{D}$ are $r_1$ and $r_2$ (possibly $r_1=r_2$).
Let $\Phi$ be an $n\times N$ matrix obtained from Construction \ref{con0} using $\mathcal{D}$ and let
$t^{\star}$ be the greatest integer $t$ such that $\Phi$ has $(\ell_{1}, t)$-recoverability.
Then $t^{\star} = c\sqrt{n}$ for some $c$ in the real interval $[\frac{1}{4},\sqrt{2}]$.
\end{theorem}

A matter of practical concern is whether, despite the presence of vectors of a certain sparsity that cannot be recovered, one can nonetheless recover many vectors of larger sparsity. We consider a sample computation.

\begin{example}\label{overperformance}
An arc of maximal size in $\PG(2,q)$ is called an oval. When $q$ is odd, a celebrated result of Segre shows that an oval is the set of points on a conic and has size $q+1$ (see Chapter 8 of \cite{Hirschfeld}). A convenient construction for ovals is as the negation of a Singer set (Proposition VII.5.12 of \cite{BJL}). The projective plane $\PG(2,7)$ corresponds to a $\BIBD(57,8,1)$. Remove eight points of an oval to produce a $\PBD(49, \{6, 7, 8\}, 1)$. Apply Construction \ref{con0} with $(V,\B)$ as the PBD using  Fourier matrices of order 8, to create a $57 \times 392$ matrix $\Phi_{\mathbb{C}}$. Apply Lemma \ref{BourgainTrick} to create a $114 \times 784$ matrix $\Phi_{\mathbb{R}}$. The spark of $\Phi_{\mathbb{C}}$ is at most $16$ by Proposition \ref{upperbound}. So, by Proposition \ref{sparkRIP}, there exist $8$-sparse vectors that cannot be recovered by $\Phi_{\mathbb{C}}$ -- this is in agreement with the bounds $ 1 \leq t \leq 10$ obtained from Theorem \ref{CSBounds} (with $n = 57$).

Figure \ref{fig1} shows the number of successful recoveries of vectors of sparsity $t$ sampled by $\Phi_{\mathbb{R}}$ out of 1,000 attempted recoveries using an OMP algorithm \cite{COSAMP} for each $t \in \{1, 2, \ldots, 50\}$. (See Section \ref{simulationsSec} for  details of all simulations.)

\begin{figure}[ht]
\centering
\includegraphics[width=15cm]{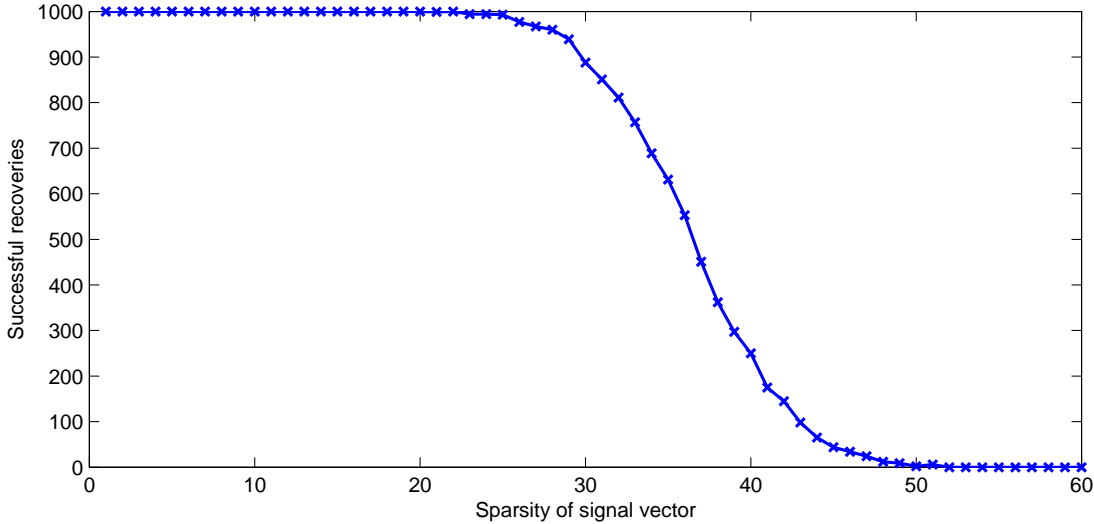}
\caption{Signal recovery for a $114 \times 784$ matrix obtained from Construction \ref{con0}}
\label{fig1}
\end{figure}

\end{example}
Example \ref{overperformance}  suggests that
although there are vectors of sparsity less than $\sqrt{2n}$ which cannot be recovered
by an $n \times N$ matrix $\Phi$ obtained from Construction \ref{con0}, such vectors are rare.
In fact, the typical performance of $\Phi$ is much better. In the next section we  explore
heuristic arguments for why this should be so.

\section{Heuristic arguments for small sparsities}\label{heuristicsSec}

Throughout this section we suppose that $\Phi$ is an $n\times N$ matrix obtained
from Construction \ref{con0} using a PBD $\mathcal{D}$ with all replication numbers equal to $r$.
As in Example \ref{overperformance}, simulations
suggest that although there are vectors of sparsity $2r$
in the nullspace of $\Phi$, and hence there exist vectors of
sparsity $r$ which cannot be recovered, such vectors are rare.
In this section, we give a heuristic justification,
inspired by techniques from random matrix theory. The
exposition in \cite{TaoRandomMatrixTheory} is relevant here.

\subsection{Random matrix theory}

Let $M$ be an $n\times n$ Hermitian matrix. Then the eigenvalues
$\lambda_1(M) \geq \cdots \geq \lambda_n(M)$ of $M$
are real and the function $L_{M}(a) = \frac{1}{n} | \{ i:
\lambda_{i}(M) \leq a\sqrt{n}\}|$
defines the density function of a (discrete) probability distribution
on $\mathbb{R}$.

One of the central problems of random matrix theory is to describe the
function $L_{M}$ when $M$ is drawn from some specific class of matrices.
The fundamental result, due to Wigner,
concerns real symmetric $n \times n$ matrices with all lower
triangular entries drawn independently from a Gaussian
$(0,1)$-distribution. His main result was that as $n \rightarrow
\infty$, $L_{M}$ converges to the
\textit{semi-circle distribution}:
\[ L(s) =  \frac{1}{2\pi} \int_{-1}^{s} \sqrt{1-s^{2}} \;ds.\]

Wigner originally demonstrated pointwise convergence and later contributors
obtained stronger results, weakening assumptions on the probability distribution
of the entries of $M$ and establishing convergence of measure. We  require only one of the bounds on the eigenvalues of $M$,
though there is ample scope for further application of random matrix theory to the analysis
of deterministic matrix constructions. The {\em operator norm} of a Hermitian matrix $M$, denoted $\op{M}$, is the maximum absolute value of an eigenvalue of $M$. Following \cite{TaoRandomMatrixTheory} we will say that an event $E$ depending on a parameter $t$ occurs  occurs {\it with high probability} if $\mathbb{P}(E)>1-o(1)$ and {\it with overwhelming probability} if, for every fixed $A>0$, $\mathbb{P}(E)>1-C_At^{-A}$ for some $C_A$ not depending on $t$.

\begin{lemma}[\cite{TaoRandomMatrixTheory}, Corollary 2.3.6]\label{Tao}
Let $M=(m_{ij})$ be a random $2t \times 2t$ Hermitian matrix such that the entries $m_{ij}$ for $1 \leq i \leq j \leq 2t$ are jointly independent with mean 0 and magnitude uniformly bounded by $1$. Then there exist absolute constants $C,c > 0$ such that for any real number $A \geq C$,
\[\mathbb{P}\left(\op{M} >A\sqrt{2t}\right) \leq C \exp(-cAt).\]
In particular, for large $t$, $\op{M}=O(\sqrt{t})$ with overwhelming probability.
\end{lemma}

\subsection{RIP bounds and signal recovery}

We recall that $\Phi$ has the \textit{restricted isometry property}
with constants $t, \delta$, abbreviated $\RIP(t, \delta)$ if and only if
for each $t$-sparse vector $m$,
\[ (1-\delta) \|m\|_{2}^{2} \leq \| \Phi m \|_{2}^{2} \leq (1+ \delta) \|m\|_{2}^{2}. \]
This is equivalent to the statement that, for each set $S$ of at most $t$ columns of $\Phi$, the eigenvalues of the matrix $\Phi_{S}^{\ast} \Phi_{S}$ all lie in
the interval $[ 1-\delta, 1+\delta]$, where $\Phi_{S}$ is the submatrix of $\Phi$
containing only the columns in $S$.

RIP conditions have been used extensively to provide sufficient conditions for $(\ell_1,t)$-recoverability \cite{CandesRombergTaoRobustUncertainty, CandesRIP}. For our purposes  it is enough to note that
if $\Phi$ satisfies the $\RIP(2t, \sqrt{2}-1)$, then $\Phi$ has $(\ell_1,t)$-recoverability. In the next subsection,
we  develop a simple model for signal recovery, which suggests that $\Phi$ recovers vectors of sparsity $O(\sqrt{n\log n})$ with high probability.

\subsection{A heuristic model}

We begin by developing a model for signals $m$ with the property that
all non-zero coordinates occur in columns of $\Phi$ corresponding to different
points of the design $\mathcal{D}$. Recall that we are assuming that all points in
$\mathcal{D}$ have equal replication number $r$.

Let $S$ be a set of $2t$ columns of $\Phi$, no pair corresponding to the same point, and
denote by $\Phi_{S}$ the submatrix of $\Phi$ consisting of the columns in $S$.
With these assumptions, $\Phi_{S}^{*} \Phi_{S}$ is of the form
$I + \Psi_{S}$ where $\Psi_{S}$ is a Hermitian matrix with zero diagonal,
and all off-diagonal entries of magnitude $\frac{1}{r}$. We lack information on
the \textbf{phase} of the entries of $\Psi_S$. Consider the space of matrices $\Psi_{S}$
as $S$ varies over all sets of $2t$ columns of $\Phi$ with no pair corresponding to the same point.
Our heuristic is that a typical element of this space behaves as the matrix $\frac{1}{r}M$, where $M$ is
a $2t \times 2t$ random Hermitian matrix with zero diagonal in which the strictly upper
triangular entries have magnitude $1$ and uniformly random phase.

Now, observe that the eigenvectors of $I + \Psi_{S}$ are those of $\Psi_{S}$. By Lemma \ref{Tao}, for large $t$, $\op{M}=O(\sqrt{t})$ with overwhelming probability. Thus, provided that $t = o(r^{2})$ as $r \rightarrow \infty$, the eigenvalues of $\Phi_{S}^* \Phi_{S}=I + \Psi_{S}$ are all arbitrarily close to $1$ with overwhelming probability. In particular, our heuristic suggests that $\Phi$ recovers \textbf{any} vector with the property that its support intersects each point of the design in at most one column with high probability.

Allowing $d$ non-zero entries in $S$ to be labelled by the same point of $\mathcal{D}$ introduces $2\binom{d}{2}$ off-diagonal zero entries in $\Psi_{S}$. The eigenvalues of a matrix are continuous functions of the matrix entries (via the characteristic polynomial). So provided that the total number of zeroes introduced is not excessive, the analysis of the heuristic continues to hold. Thus, it seems reasonable to suppose that our heuristic is valid for signals of sparsity at most $r \log r$ with at most $\log r$ non-zero coordinates corresponding to any point of the design. We  call such signal vectors \textit{suitable}.

What we have obtained so far is not a uniform recovery guarantee, however; it does \textbf{not} follow that $\Phi$ recovers \textbf{all} suitable vectors with high probability. We need a slightly more careful analysis to obtain a uniform recovery guarantee. Taking $t = r \log r$ and $A = \frac{1}{2}(2-\sqrt{2}) r^{1/2}(\log r)^{-1/2}$ in Lemma \ref{Tao}, we obtain
\[ \mathbb{P}\left(\tfrac{1}{r}\op{M} > \sqrt{2}-1\right) \leq Ce^{-c r^{3/2}(\log r)^{1/2}} \]
for absolute constants $C$ and $c$. Now, the total number of subsets of size $2r\log r$ of a set of size $N = vr = O(r^{3})$ is bounded above by $N^{2r\log r} = O(e^{6r (\log r)^{2}})$. So, because $r^{3/2}(\log r)^{1/2}$ grows faster than $6r (\log r)^{2}$, we have that $\mathbb{P}(\tfrac{1}{r}\op{M} > \sqrt{2}-1)=o(N^{-2r\log r})$.
It follows by the union bound that for sufficiently large $r$ the eigenvalues of $\frac{1}{r}M$ all lie in the interval $[2-\sqrt{2}, \sqrt{2}]$ with high probability. (Recall that $r$ is the replication number of the design used in Construction \ref{con0} to obtain $\Phi$ and that the order of $M$ is a function of $r$.)

Consequently, the heuristic suggests that there is a high probability that, for large $r$, $\Phi$ satisfies the $\RIP(2r\log r, \sqrt{2}-1)$ on all suitable sets of columns, and so all suitable vectors are recoverable. By Proposition \ref{CSUpperBound} and Lemma \ref{sparkRIP} we know that there exist vectors of sparsity $r$ that are not recoverable, but these are far from suitable.

This analysis suggests that for large $r$, $\Phi$  recovers all suitable vectors of sparsity $\sqrt{n\log n}$, with high probability. This result is in excess of the square-root bound. The probability that $\Phi$ fails to recover any particular random signal vector is exponentially small, being essentially the probability that many non-zero coordinates of the signal are concentrated in columns labelled by a small number of points.

\section{Simulations}\label{simulationsSec}

In this section we describe the results of extensive simulations. All simulations are
performed with real-valued matrices and signal vectors. If our construction yields a complex matrix, we apply Lemma \ref{BourgainTrick} to obtain a real one. To illustrate the flexibility of Construction \ref{con0} we derive input PBDs from BIBDs in various ways.

Our matrices are stored as an array of real floating point numbers of some
fixed accuracy $l$. The entries typically consist of numbers of the form $\cos(2s\pi/r)$ or
$\sin(2s\pi/r)$ where $r$ is a replication number of the design used to create the matrix and $s \in \{0,\ldots,r-1\}$.
Our simulations are performed as follows.

\begin{itemize}

	\item Vectors of sparsity $t$ are constructed by choosing $t$ coordinate positions uniformly at random,
and populating these locations either with a value from a uniform distribution on $(0, 1)$ (in matlab)
or a uniform distribution on $\{ \frac{i}{100}: 1 \leq i \leq 100\}$ (in MAGMA) and then normalising. In all discussion of simulations, when we refer to a $t$-sparse vector we mean that it has exactly $t$ non-zero coordinates.
	
	\item We invoke the standard implementation of one of the recovery algorithms from Section \ref{sec:ra} to obtain a vector $\hat{m}$ such that $\Phi \hat{m} \approx \Phi m$.
	Only $\Phi$ and $\Phi m$ are provided to the solver;   the sparsity of $m$ is not supplied.
	
	\item The other parameters in our test are the precision to which real numbers are stored (typically 25 decimal places) and the recovery error allowed, $\epsilon$. A trial is counted a success if $|m-\hat{m}| < \epsilon$, and a failure otherwise. We report the proportion of successes at sparsity $t$ as a proxy for compressed sensing performance.
\end{itemize}

Some general comments  apply to all of our simulations.
In essence choosing a design and Hadamard matrices for the construction is a multi-dimensional optimisation, probably application specific. After fixing a matrix $\Phi$, one chooses the precision, the maximal entry size in a random vector and the error allowed in recovery. These all impact  performance, both proportion of correct recoveries and recovery time.

The recovery process is not very sensitive to the allowable recovery error, in the sense
that if the algorithm converges to the correct solution it generally does so to machine tolerance.
Taking sparse vectors to be binary valued greatly improves recovery, but we focus  on the model given. Numerical instability results if the precision to which real numbers are stored is too small; this is determined experimentally.

\subsection{Recovery algorithms}\label{sec:ra}

We compare the performance of two different well-established recovery methods when applied to signal vectors that have been sampled using matrices created via Construction \ref{con0}.

Firstly, we employ na\"{\i}ve Linear Programming (LP); we consider implementations in both \textsc{MAGMA} and \texttt{matlab}, \cite{MAGMA, matlab}. We include implementations in two different systems to demonstrate the potential differences between solvers, which can be significant.

Secondly, we employ \textit{matching pursuit}, which is a greedy algorithm for producing sparse approximations of high dimensional data in an over-complete dictionary. A well-known implementation for compressed sensing is CoSaMP \cite{COSAMP}. We use the implementation of the basic OMP algorithm from \cite{COSAMP} rather than the more specialised CoSaMP algorithm.
It is well known that the worst case complexity of the simplex method in linear programming is exponential in the number of variables.
However the expected complexity for typical examples is $O(N^3)$. In comparison, the complexity of CoSaMP is $O(N \log^2N)$.

We created a $57 \times 456$ matrix $\Phi_{\mathbb{C}}$ using Construction \ref{con0}, taking $(V,\mathcal{B})$ to be a $\BIBD(57, 8, 1)$ corresponding to a $\PG(2,7)$ and $H_1,\ldots,H_v$ to be Fourier matrices of order $8$. We then applied Lemma \ref{BourgainTrick} to obtain a $114 \times 912$ matrix $\Phi_{\mathbb{R}}$. The allowable recovery error was $10^{-8}$ for all algorithms. In Figure \ref{fig2}, we record the number of successful recoveries for $1000$ randomly generated vectors sampled by $\Phi_{\mathbb{R}}$ for each algorithm and each sparsity.

Applying Theorem \ref{CSBounds} shows that $\Phi_{\mathbb{C}}$ allows recovery of all vectors of sparsity at most $\lceil\frac{1}{4}\sqrt{57}\rceil=2$, and provides no stronger guarantee. Furthermore, $\Phi_{\mathbb{C}}$ contains real rows so, as in Remark \ref{counter}, there exist
$12$-sparse vectors in the nullspace of $\Phi_{\mathbb{C}}$, and hence there exist $6$-sparse
vectors whose recovery is not permitted by $\Phi_{\mathbb{C}}$. So the recovery performance of $\Phi_{\mathbb{R}}$ is quite striking.

\begin{figure}[ht]
\centering
\includegraphics[width=15cm]{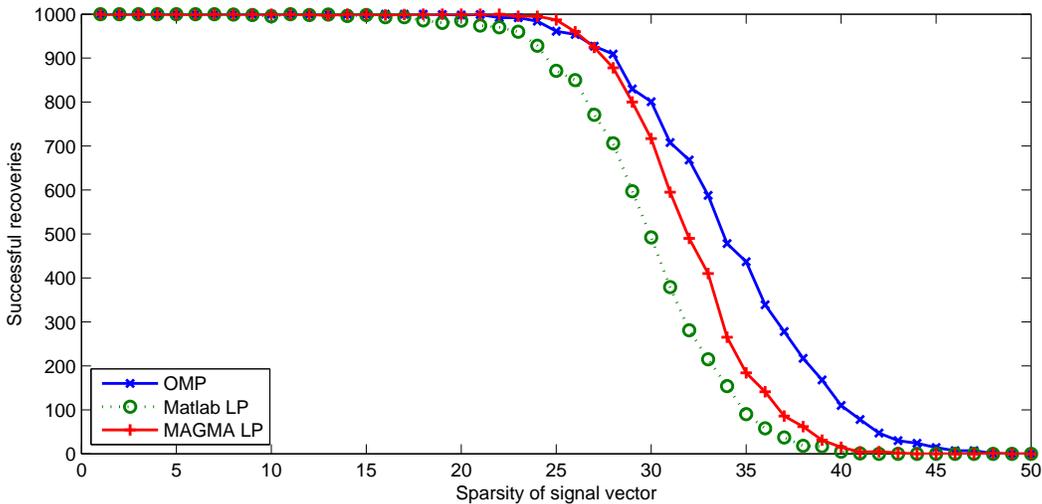}
\caption{Comparison of recovery algorithms}
\label{fig2}
\end{figure}

While there are clear differences in performance among these algorithms,
they appear to behave in a broadly similar fashion. In general, running times for
OMP are an order of magnitude faster than the matlab linear programming solver. The \textsc{MAGMA} solver
has intermediate runtime.

\subsection{Comparison with Gaussian matrices}

Gaussian matrices are the de facto standard against which other matrix constructions are measured in compressed sensing. The projective plane $\PG(2, 11)$ corresponds to a $\BIBD(133, 12, 1)$. Removing two blocks and all points incident with either block produces a $\PBD(110, \{10, 11\}, 1)$ in which all points have replication number $12$. We applied Construction \ref{con0}, taking $(V,\mathcal{B})$ to be this PBD and $H_1,\ldots,H_v$ to be Fourier matrices of order 12, to obtain a $131 \times 1320$ matrix $\Phi_{\mathbb{C}}$. We then applied Lemma \ref{BourgainTrick} to obtain a $262 \times 2640$ matrix $\Phi_{\mathbb{R}}$.

We compare the compressed sensing recovery performance of $\Phi_{\mathbb{R}}$ to that of a Gaussian ensemble with the same numbers of rows and columns (using the OMP algorithm). Both linear programming and OMP ran faster by an order of magnitude for $\Phi_{\mathbb{R}}$ than for the Gaussian ensemble. Our results are given in Figure \ref{fig3}.

\begin{figure}[ht]
\centering
\includegraphics[width=15cm]{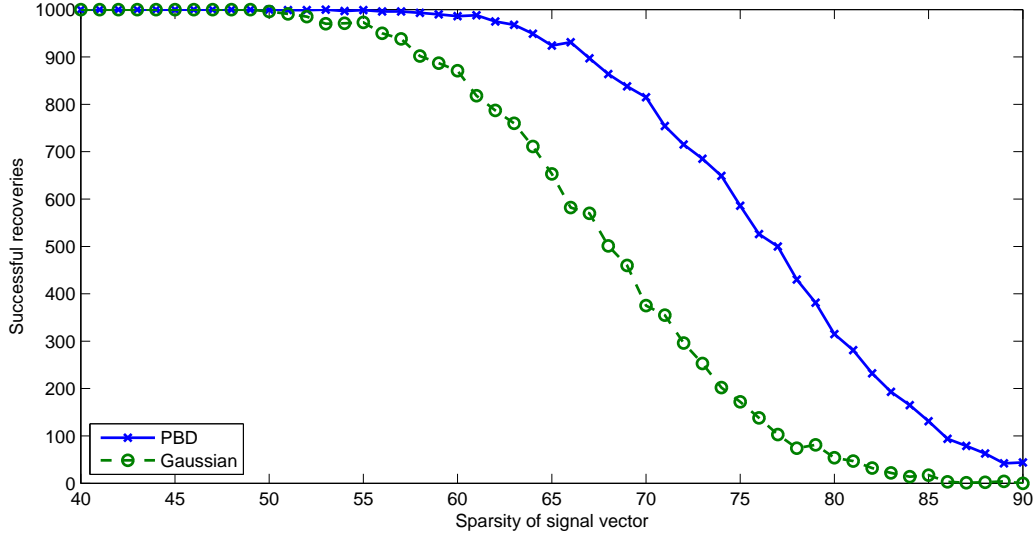}
\caption{Comparison with Gaussian ensemble}
\label{fig3}
\end{figure}

\subsection{Factors affecting algorithm performance}

In this section we discuss a number of factors that influence recovery performance. Specifically we consider the presence of noise in received signals,
signed signal vectors (until now our signal vectors have had positive coordinates),
and the effect of using different Hadamard matrices in  Construction \ref{con0}.
In all cases, we find that the construction is robust and reductions in performance are not substantial.

\subsubsection{Signed signal vectors}

The linear programming solver in \textsc{MAGMA} requires variables to be positive
(or at least greater than some bound). We use a standard trick to allow negative entries
in $x$. For each variable $x_{i}$ in the original problem, we introduce a pair of variables,
$x_{i}^{+}$ and $x_{i}^{-}$. Then we replace each appearance of $x_{i}$ with $x_{i}^{+} - x_{i}^{-}$.

When a solution has been found, we interpret $x_{i}^{+}$ as a positive number, and $x_{i}^{-}$
as a negative one. This has the disadvantage of doubling the number of variables in the linear program. This results in longer run times and slightly lower performance, but the results of the simulations are broadly comparable to those in the positive case.

\subsubsection{Noise}

We consider both uniform and burst noise. In each case we consider positive
and signed noise vectors. In general, recovery in the presence of noise is robust.
Testing all algorithms in all regimes would produce an overabundance of data, so
we give only a single representative example of our simulations.

The projective plane $\PG(2,11)$ corresponds to a $\BIBD(133,12,1)$. Remove twelve points of an oval to produce a $\PBD(121, \{10, 11, 12\}, 1)$. Apply Construction \ref{con0} taking $H_1,\ldots,H_v$ to be Fourier matrices of order 12, and then apply Lemma \ref{BourgainTrick} to construct a compressed sensing matrix $\Phi_{\mathbb{R}}$ with $266$ rows and $2904$ columns. To examine the performance of $\Phi_{\mathbb{R}}$ in the presence of uniform positive noise, we construct noise vectors with entries uniformly distributed in $(0,1)$ and scale the vector to some predetermined $\ell_2$-norm. We then compare the signal vector $m$ (of $\ell_2$-norm $1$) to the solution returned by the matlab linear programming solver for $\Phi_{\mathbb{R}}(m+\epsilon)$. We count a recovery as a success if the reconstruction error is below $10^{-8}$. Table \ref{tab1} summarises our results.

\begin{table}[H]
\begin{center}
\begin{tabular}{|c||c|c|c|c|c|}
\hline
&\multicolumn{5}{c|}{$\ell_2$-norm of noise vector}  \\
\cline{2-6}
Sparsity & 0 & $10^{-12}$ & $10^{-10}$ & $10^{-9}$ & $2\times 10^{-9}$ \\
\hline
30 & 100 &  99 & 98 & 79 & 66 \\
35 & 100 & 100 & 97 & 79 & 69 \\
40 & 100 & 100 & 91 & 77 & 49 \\
45 &  97 &  93 & 88 & 62 & 27 \\
50 &  87 &  79 & 69 & 33 &  5 \\
55 &  61 &  56 & 30 & 14 &  2 \\
60 &  27 &  22 & 22 &  5 &  0 \\
\hline
\end{tabular}
\caption{Number of successful recoveries out of 100 for signals of various sparsities and for different noise levels.}
\label{tab1}
\end{center}
\end{table}

Recovery decays gracefully in the presence of noise, particularly when the sparsity of the signal is not close to the limit of what can be recovered.

\subsubsection{Choice of Hadamard matrix}

We illustrate the effect of the choice of Hadamard matrix with a small example. We take a $\BIBD(25,3,1)$ obtained from the Bose construction, which has replication number $r = 12$ (see p.25 of \cite{ColbournRosa}). The $100 \times 300$ matrices $\Phi_{\mathbb{C}}$ and $\Phi'_{\mathbb{C}}$ are obtained via Construction \ref{con0}, taking $(V,\B)$ to be this BIBD and $H_1,\ldots,H_v$ to be real Hadamard matrices of order 12 and Fourier matrices of order 12, respectively. Lemma \ref{BourgainTrick} is then applied to both matrices to obtain $200 \times 600$ matrices $\Phi_{\mathbb{R}}$ and $\Phi'_{\mathbb{R}}$. (While $\Phi_{\mathbb{C}}$ is already real, this application of Lemma \ref{BourgainTrick} allows a direct comparison between constructions. It has no effect on the sparsity of vectors recovered.) Recovery performance varies by less than $2 \%$ and runtime by less than $6\%$ with the OMP algorithm. Such variation could be caused simply by random fluctuations. On the other hand, the differences in performance in the MAGMA LP-implementation are substantial. We record them in Table \ref{tab2}.

\begin{table}[H]
\begin{center}
\begin{tabular}{|c||c|c||c|c|}
\hline
&\multicolumn{2}{c||}{Real matrix}&\multicolumn{2}{c|}{Fourier matrix} \\ \cline{2-5}
Sparsity & No. successes & Avg. time & No. successes & Avg. time \\
\hline
56 & 100 & 4.8 & 100& (crashed twice)\\
58 & 100 & 4.5 & 99 & 14.5\\
60 & 98  & 4.4 & 96 & 14.8\\
62 & 100 & 4.6 & 98 & 15.1\\
64 & 93  & 4.8 & 99 & 15.6 \\
66 & 96  & 4.8 & 97 & 15.9\\
68 & 94  & 4.9 & 96 & 16.3\\
70 & 87  & 5.0 & 99 & 16.8\\
\hline
\end{tabular}
\caption{Number of successful recoveries out of 100 and the average recovery time in seconds for real and Fourier matrices.}
\label{tab2}
\end{center}
\end{table}

With the MAGMA LP-implementation, it would appear that $\Phi'_{\mathbb{R}}$ produces
better recovery at the cost of increased runtime and the risk of numerical instability. Obviously the choice of matrix and recovery algorithm would depend on the particular application.

\section{An efficient algorithm for sparse recovery}\label{NewAlgSec}

In this section we describe and investigate a new algorithm for signal recovery, tailored specifically for an $n \times N$ compressed sensing matrix created via Construction \ref{con0}. It is designed to recover vectors of sparsity at most $O(\sqrt{n})$ and is not expected to be competitive with LP or CoSaMP at large sparsities. The algorithm exploits the structure of matrices created via Construction \ref{con0} in order to achieve efficiency in both running time and storage requirements. Under certain assumptions, it can be shown to run successfully in time $O(N\log n)$ and space $O(n^2)$. (Such thresholding based algorithms are one of the main algorithmic approaches to compressed sensing; see Chapter 3 of \cite{FoucartRauhut}.) Throughout this section we  employ the notation that we introduce in describing the algorithm.

\begin{algorithm}\label{NewAlg}
Suppose that $\Phi$ is a matrix created via Construction \ref{con0}, using a design $(V,\mathcal{B})$ with replication numbers $r_1 \leq \cdots \leq r_v$ and Hadamard matrices $H_1,\ldots,H_v$. For each $i \in V$, let $\Phi_i$ be the submatrix of $\Phi$ consisting of the columns corresponding to point $i$. We store only the following information.
\begin{itemize}
    \item
For each $i \in V$, a list of the rows of $\Phi_i$ that are non-zero, and a record of which rows of $H_{i}$ are located there.
    \item A copy of each Hadamard matrix $H_{i}$.
\end{itemize}
Suppose that $\Phi$ is used to sample some $N \times 1$ message vector $m$. Our algorithm runs as follows.
\begin{enumerate}
    \item
Construct an initial estimate. For each $i \in V$, we take the set of rows in which $\Phi_i$ is non-zero, take the $r_i \times 1$ vector $y_i$ of the corresponding received samples, and compute $\hat{m}_i=\frac{1}{\sqrt{r_i}}H_{i}^{\ast}y_{i}$. Concatenating the vectors $\hat{m}_i$ over all $i \in V$ we construct an initial estimate $\hat{m}$ for the signal vector $m$.
    \item
Guess the signal coordinates. For some prespecified $|S| \leq r_1$, let $S$ be the index set of the $|S|$ coordinates of $\hat{m}$ of greatest magnitude. For each $i \in V$, let $q_i$ be the number of columns indexed by $S$ that are in $\Phi_i$, and let $V'=\{i \in V:q_i \geq 1\}$.
    \item
Identify uncontaminated samples. For each $i \in V'$, we find a set $Q_i$ of $q_i$ rows of $\Phi$ that are non-zero in $\Phi_i$ but zero in $\Phi_j$ for each $j \in V'\setminus\{i\}$ (such rows correspond to blocks that contain $i$ but no point in $V'\setminus \{i\}$ and, because $|V'\setminus \{i\}| \leq |S|-q_i \leq r_i-q_i$, at least $q_i$ such blocks exist).
    \item
Recover the signal. For each $i \in V'$, we find the coordinates in the $q_i$ positions indexed by $S$ that correspond to columns in $\Phi_i$ by solving the $q_i \times q_i$ linear system induced by those columns and the rows in $Q_i$.
\end{enumerate}
\end{algorithm}

\begin{remark}
The record of the non-zero rows of $\Phi$ can be derived easily from the incidence matrix of $(V,\mathcal{B})$. For many designs (for example, projective planes, designs with cyclic automorphisms, and so on) this information need not be stored explicitly. Even if the information is stored explicitly, the space used is $O(n^{2})$. A similar observation holds for the Hadamard matrices; Fourier or Paley matrices need not be stored explicitly. With appropriate design choices, the storage space required can be logarithmic in $n$.
\end{remark}

As long as $S$ contains the positions of all the non-zero elements, the algorithm finds the right solution.
In our analysis of this algorithm we focus, for the sake of simplicity, on the case in which all the replication numbers of $(V,\mathcal{B})$ are equal. We first show that our estimate for a signal coordinate differs from the actual value by an error term introduced by the other non-zero signal coordinates. From this it is easy to show that Algorithm \ref{NewAlg} recovers a signal provided that no non-zero signal coordinate has magnitude too small in comparison with the $\ell_1$-norm of the signal.

\begin{lemma}\label{algCoord}
Let $\Phi$ be a matrix created via Construction \ref{con0} using a PBD $\mathcal{D}$ all of whose replication numbers are equal to some integer $r$. Suppose that $\Phi$ is used to sample a signal $m$ of sparsity at most $r$ and that Algorithm \ref{NewAlg} is applied to find an estimate $\hat{m}$ for $m$. Let $\hat{x}$ be the estimate for a coordinate $x$ of $m$ that corresponds to a point $i$ of $\mathcal{D}$. Then
\begin{itemize}
    \item[(i)]
$\hat{x}=x+\frac{1}{r}(h_1x_1+\cdots +h_dx_d)$ where $x_1,\ldots,x_d$ are the non-zero coordinates of $m$ that do not correspond to point $i$ of $\mathcal{D}$ and $h_1,\ldots,h_d$ are complex numbers of magnitude $1$; and
    \item[(ii)]
$|\hat{x}-x| \leq \frac{1}{r}\|m\|_1$.
\end{itemize}
\end{lemma}

\begin{proof}
Let $x_1,\ldots,x_d$ be the non-zero coordinates of $m$ that do not correspond to point $i$ of $\mathcal{D}$ and let $w_i$ be the vector $(x_1,\ldots,x_d)^{\top}$. Let $m_i$ be the $r \times 1$ vector consisting of the coordinates of $m$ corresponding to point $i$ of $\mathcal{D}$ and let $\hat{m}_i$ be the estimate for $m_i$. Note that $y_i=\frac{1}{\sqrt{r}}H_im_i+A_iw_i$ for some $r \times d$ matrix $A_i$ such that each column of $A_i$ contains exactly one complex number of magnitude 1 and every other entry of $A_i$ is a 0. So
$$\hat{m}_i=\tfrac{1}{r}H^{\ast}_iH_im_i+\tfrac{1}{\sqrt{r}}H^{\ast}_iA_iw_i=m_i+\tfrac{1}{\sqrt{r}}(H^{\ast}_iA_i)w_i.$$
Thus, because every entry of $H^{\ast}_iA_i$ is a complex number of magnitude 1,
$$\hat{x}=x+\tfrac{1}{r}(h_1x_1+\cdots +h_dx_d)$$
where $h_1,\ldots,h_d$ are complex numbers of magnitude $1$. Further,
$$|\hat{x}-x|=|\tfrac{1}{r}(h_1x_1+\cdots +h_dx_d)|\leq\tfrac{1}{r}(|x_1|+\cdots+|x_d|)\leq\tfrac{1}{r}\|m\|_1.$$
\end{proof}

\begin{corollary}\label{recovCor}
Let $\Phi$ be a matrix created via Construction \ref{con0} using a PBD $\mathcal{D}$ all of whose replication numbers are equal to some integer $r$. Suppose that $\Phi$ is used to sample a signal $m$ of sparsity $t \leq r$. If each non-zero coordinate of $m$ has magnitude at least $\frac{2}{r}\|m\|_1$, then Algorithm \ref{NewAlg} with $|S| \geq t$ recovers $m$.
\end{corollary}

\begin{proof}
Let $\hat{x}_0$ and $\hat{x}_1$ be the estimates obtained by Algorithm \ref{NewAlg} for two coordinates $x_0$ and $x_1$ of the signal such that $x_0=0$ and $x_1\neq0$. It suffices to show that $|\hat{x}_0| < |\hat{x}_1|$. By Lemma \ref{algCoord}(ii), $|\hat{x}_0|<\frac{1}{r}\|m\|_1$. From our hypotheses $|x_1|\geq\frac{2}{r}\|m\|_1$ and so by Lemma \ref{algCoord}(ii), $|\hat{x}_1|>|x_1|-\frac{1}{r}\|m\|_1>\frac{1}{r}\|m\|_1$. Thus, $|\hat{x}_0|<|\hat{x}_1|$.
\end{proof}

If we assume that the non-zero signal coordinates have uniformly random phase,  we can improve substantially on Corollary \ref{recovCor}. Also, if we assume that the support of the signal is chosen uniformly at random,  we can establish that, for large $n$, the algorithm runs in $O(N\log n)$ time with high probability. We employ a simple consequence of Hoeffding's inequality.

\begin{lemma}\label{HoefLemma}
If $z_1,\ldots,z_n$ are independent complex random variables with uniformly random phase and magnitude at most $1$, then for each positive real number $c$,
$$\mathbb{P}(|z_1+\cdots+z_n| \geq c) \leq 4\exp\left(\tfrac{-c^2}{4n}\right).$$
\end{lemma}

\begin{proof}
If $|z_1+\cdots+z_n| \geq c$,  one of the real or imaginary parts of $z_1+\cdots+z_n$ must have magnitude at least $\frac{c}{\sqrt{2}}$. Now apply Hoeffding's inequality separately to the real and imaginary parts of $z_1+\cdots+z_n$, noting that the expected value is zero in each case.
\end{proof}

\begin{theorem}
Let $\Phi$ be an $n \times N$ matrix created via Construction \ref{con0} using a PBD $\mathcal{D}$ all of whose replication numbers are equal to some integer $r$. Suppose that $\Phi$ is used to sample an $r$-sparse signal $m$.
\begin{itemize}
    \item[(i)]
If the non-zero components of $m$ are independent random complex variables with uniformly random phase and magnitude in the interval $[r^{-\frac{1}{2}+\epsilon},1]$ for a positive constant $\epsilon$, then for large $n$ Algorithm \ref{NewAlg} with $|S| = r$ recovers $m$ with high probability.
    \item[(ii)]
If the support of $m$ is chosen uniformly at random, then for large $n$ the algorithm runs in $O(N\log n)$ time with high probability.
\end{itemize}
\end{theorem}

\begin{proof}
Recall that $n \sim r^2$ and suppose that $r$ is large.

We first prove (i). For brevity, let $a=r^{-\frac{1}{2}+\epsilon}$. Suppose that the non-zero components of $m$ are independent random complex variables with uniformly random phase and magnitude in the interval $[a,1]$.  Let $x$ be a coordinate of $m$ and let $\hat{x}$ be our estimate for $x$. By Lemma \ref{algCoord}(i), $\hat{x}-x=\frac{1}{r}(h_1x_1+\cdots +h_dx_d)$ where $x_1,\ldots,x_d$ are the non-zero coordinates of $m$ that do not correspond to point $i$ of $\mathcal{D}$ and $h_1,\ldots,h_d$ are complex numbers of magnitude $1$. Because $x_1,\ldots,x_d$ are independent random complex variables with uniformly random phase and magnitude at most 1, so are $h_1x_1,\ldots,h_dx_d$. Thus, by Lemma \ref{HoefLemma},
$$\mathbb{P}\left(|h_1x_1+\cdots +h_dx_d| \geq \tfrac{ar}{2}\right) \leq 4\exp\left(\frac{-a^2r^2}{16d}\right)=4\exp\left(\frac{-r^{1+2\epsilon}}{16d}\right).$$
Thus, using the facts that $d \leq r$ and that $r^{2\epsilon}$ grows faster than $\log N$ (by Remark \ref{growthRemark}, $N=O(r^3)$), it can be seen that $\mathbb{P}(|\hat{x}-x| \geq \tfrac{a}{2}) = o(\frac{1}{N})$. So it follows from the union bound that with high probability it is the case that $|\hat{x}-x| < \frac{a}{2}$ for each coordinate $x$ of $m$. Because $|x| \geq a$ for each non-zero coordinate $x$ of $m$, this implies that our estimates for non-zero coordinates of $m$ have magnitude greater than our estimates for zero coordinates of $m$ with high probability. Then Algorithm \ref{NewAlg} with $|S|=r$ recovers $m$.

We now prove (ii). Note that $|S| = r$, that $rv=N$ and that $r \sim \sqrt{n}$ (see Remark \ref{growthRemark}). The first step of the algorithm is essentially performing a Hadamard transform of order $r$ for each point, and so can be accomplished in $O(vr \log r)=O(N \log n)$ time. The second step is essentially a sorting operation and can be accomplished in $O(N)$ time. The third step can be accomplished by first creating a list of blocks that intersect exactly one point in $V'$ (by examining $r$ blocks for each of the at most $|S|$ points in $V'$) and then partitioning this list into the sets $Q_i$. Because $|S| = r \sim \sqrt{n}$, this takes $O(n)$ time. The final step involves inverting a $q_i \times q_i$ matrix for each $i \in V'$.

It is known that if $r$ balls are placed in $r$ bins uniformly at random,  the maximum number of balls in any bin is $\frac{\log r}{\log\log r}(1+o(1))$ with high probability (see \cite{RaabBalls}, for example). In this case $q_i = o(\log r)$ for each point $i$, and hence the inversions take $o(r \log^3 r)=o(n)$ time. Combining these facts, the algorithm runs in $O(N\log n)$ time with high probability.
\end{proof}

\section*{Acknowledgements}

The authors acknowledge the support of the Australian Research Council via grant DP120103067.
\bibliographystyle{abbrv}
\flushleft{
\bibliography{NewBiblio}
}

\end{document}